\documentclass[aps,pra,twocolumn,showpacs]{revtex4}
\usepackage{amsthm}
\usepackage{amsmath}
\usepackage{latexsym}
\usepackage{amsfonts}
\usepackage{amssymb}
\usepackage{color}
\usepackage{bbm,dsfont}
\usepackage{graphicx}
\usepackage{subfigure}
\usepackage{mathrsfs}
\usepackage{mathbbol}
\usepackage{hyperref}
\usepackage{enumerate}
\usepackage{MnSymbol}
\usepackage{bbding}

\newtheorem{proposition}{Proposition}
\newtheorem{proposition?}{Proposition?}

\theoremstyle{definition}

\newtheorem{definition}{Definition}

\newcommand{\integer}{\mathbb Z} 
\newcommand{\half}{\tfrac{1}{2}} 

\newcommand{\hi}{\mathcal{H}} 
\newcommand{\hik}{\mathcal{K}} 
\newcommand{\lh}{\mathcal{L(H)}} 
\newcommand{\sh}{\mathcal{S(H)}} 
\newcommand{\ip}[2]{\left\langle\,#1\,|\,#2\,\right\rangle} 
\newcommand{\kb}[2]{|#1\rangle\langle#2|} 
\newcommand{\tr}[1]{\mathrm{tr}\left[#1\right]} 
\newcommand{\ptr}[2]{\mathrm{tr}_{#1}[#2]} 
\newcommand{\id}{\mathbbm{1}} 

\newcommand{\obsh}{\mathcal{O}(\hi)}
\newcommand{\A}{\mathsf{A}}
\newcommand{\B}{\mathsf{B}}
\newcommand{\G}{\mathsf{G}}

\newcommand{\J}{\mathsf{J}}
\newcommand{\T}{\mathsf{T}}

\newcommand{\I}{\mathcal{I}}

\newcommand{\hout}{\hi'} 
\newcommand{\state}{\mathcal{S}(\hi)}
\newcommand{\stateout}{\mathcal{S}(\hi')}

\newcommand{\stateab}{\mathcal{S}(\hi_\mathcal{A} \otimes \hi_\mathcal{B})}
\newcommand{\statehh}{\mathcal{S}(\hi \otimes \hi)}
\newcommand{\statekh}{\mathcal{S}(\hik \otimes \hi)}
\newcommand{\statekk}{\mathcal{S}(\hik_1 \otimes \hik_2)}

\begin{document}

\title[]{Simultaneous Measurement of Two Quantum Observables: \\
Compatibility, Broadcasting, and In-between}

\author{Teiko Heinosaari}
\email{teiko.heinosaari@utu.fi}
\address{Turku Centre for Quantum Physics, Department of Physics and Astronomy, University of Turku, Finland}


\begin{abstract} 
One of the central features of quantum theory is that there are pairs of quantum observables that cannot be measured simultaneously. 
This incompatibility of quantum observables is a necessary ingredient in several quantum phenomena, such as measurement uncertainty relations, violation of Bell inequalities and steering. 
Two quantum observables that admit a simultaneous measurement are, in this respect, classical.
A finer classification of classicality can be made by formulating four symmetric relations on the set of observables that are stronger than compatibility; they are broadcastability, one-side broadcastability, mutual nondisturbance and nondisturbance.
It is proven that the five relations form a hierarchy, and their differences in terms of the required devices needed in a simultaneous measurement is explained.
All the four relations stronger than compatibility are completely characterized in the case of qubit observables. 
\end{abstract}

\maketitle


\section{Introduction}

It is one of the central features of quantum theory that only some pairs of quantum observables can be measured simultaneously. 
There are various ways how two quantum observables may permit a simultaneous measurement.  
Joint measurability, or compatibility, is the general concept related to simultaneous measurements.
Compatibility of two observables does not say anything how those observables can be implemented jointly, just that there is some measurement set-up giving the correct marginal probability distributions.

In contrast, broadcasting of observables is a modification of broadcasting of states, and it is a very specific way to implement simultaneous measurement.
It requires an existence of a broadcasting channel that gives two approximate copies of an arbitrary input state and, even if the copies are not identical to the original state, there is no difference with respect to the target observables.
A broadcastable pair of observables is compatible, but it is compatible in a very strong sense.

These two scenarios raise some immediate questions. 
What is exactly the additional feature that makes some compatible pairs broadcastable, especially from the point of view of implementation of their simultaneous measurement?
How different are these two relations on observables, and are there any intermediate steps between them?
In this paper we tackle these questions.

We will define three relations on observables that are between broadcastability and compatibility; they are weaker than broadcastability but stronger than compatibility.
These relations are \emph{one-side broadcastability}, \emph{mutual nondisturbance} and \emph{nondisturbance}.
All together we then obtain a hierarchy of five relations on quantum observables; see Fig. \ref{fig:hierarchy}.

The hierarchy of relations is useful in several different ways. 
Firstly, it reveals that there are different levels of joint measurability, and in this sense, different layers of classicality.
Secondly, if we can show that some pair of observables is, e.g., not compatible, then we know that all the stronger relations fail as well.
We will demonstrate the usage of this kind of argument, and we will completely characterize all the four relations stronger than compatibility in the case of qubit observables.

To understand the differences of the five relations, we will formulate them in a unifying way. 
We show that the five relations can be understood in the differences of the needed devices in the implementation of a simultaneous measurement.
Using the presented framework we can also demonstrate that a natural generalization of the compatibility relation is, in fact, equivalent to the compatibility.

\section{Compatibility}

In the following $\hi$ is a fixed Hilbert space, either finite or countably infinite dimensional.
We denote by $\sh$ the set of all states, i.e., positive trace class operators of trace 1.
A quantum observable is mathematically defined as a positive operator valued measure (POVM) \cite{PSAQT82}, \cite{OQP97}.
We will restrict our investigation to observables with finite number of outcomes, hence we will understand an observable as a map $\A$ from a finite set of measurement outcomes $\Omega_\A$ to the set of bounded linear operators $\lh$ on $\hi$ such that $\A(x)\geq 0$ and $\sum_x \A(x) = \id$.
For a subset $X\subseteq\Omega_\A$, we denote $\A(X) = \sum_{x\in X} \A(x)$.
The probability of getting an outcome $x$ in a measurement of $\A$ in an initial state $\varrho$ is given by the formula $\tr{\varrho \A(x)}$.

We denote by $\obsh$ the set of all observables $\A$ on $\hi$ with $\Omega_\A \subset \integer$.
A binary relation on $\obsh$ is a subset $\mathcal{R}$ of the Cartesian product $\obsh \times \obsh$.
Hence, a binary relation can be thought as a property that a pair of observables may or may not possess.
The relations that we will study are all symmetric:
if $(\A,\B) \in \mathcal{R}$, then also $(\B,\A) \in \mathcal{R}$.
For this reason, we can talk about properties of $\A$ and $\B$ rather than $(\A,\B)$.
If $\mathcal{R}$ and $\mathcal{R}'$ are two binary relations on $\obsh$ such that  $\mathcal{R} \subset \mathcal{R}'$, then we say that  $\mathcal{R}$ is \emph{stronger} than $\mathcal{R}'$, and that $\mathcal{R}'$ is \emph{weaker} than  $\mathcal{R}$.
The complement of a binary relation $\mathcal{R}$ is the subset of those pairs $(\A,\B) \in \obsh \times \obsh$ that do not belong to $\mathcal{R}$.
The complement relation of a symmetric relation is also symmetric, and the inclusion of two relations is reversed in their complement relations.

\begin{figure}
\begin{center}
\includegraphics[width=6.5cm]{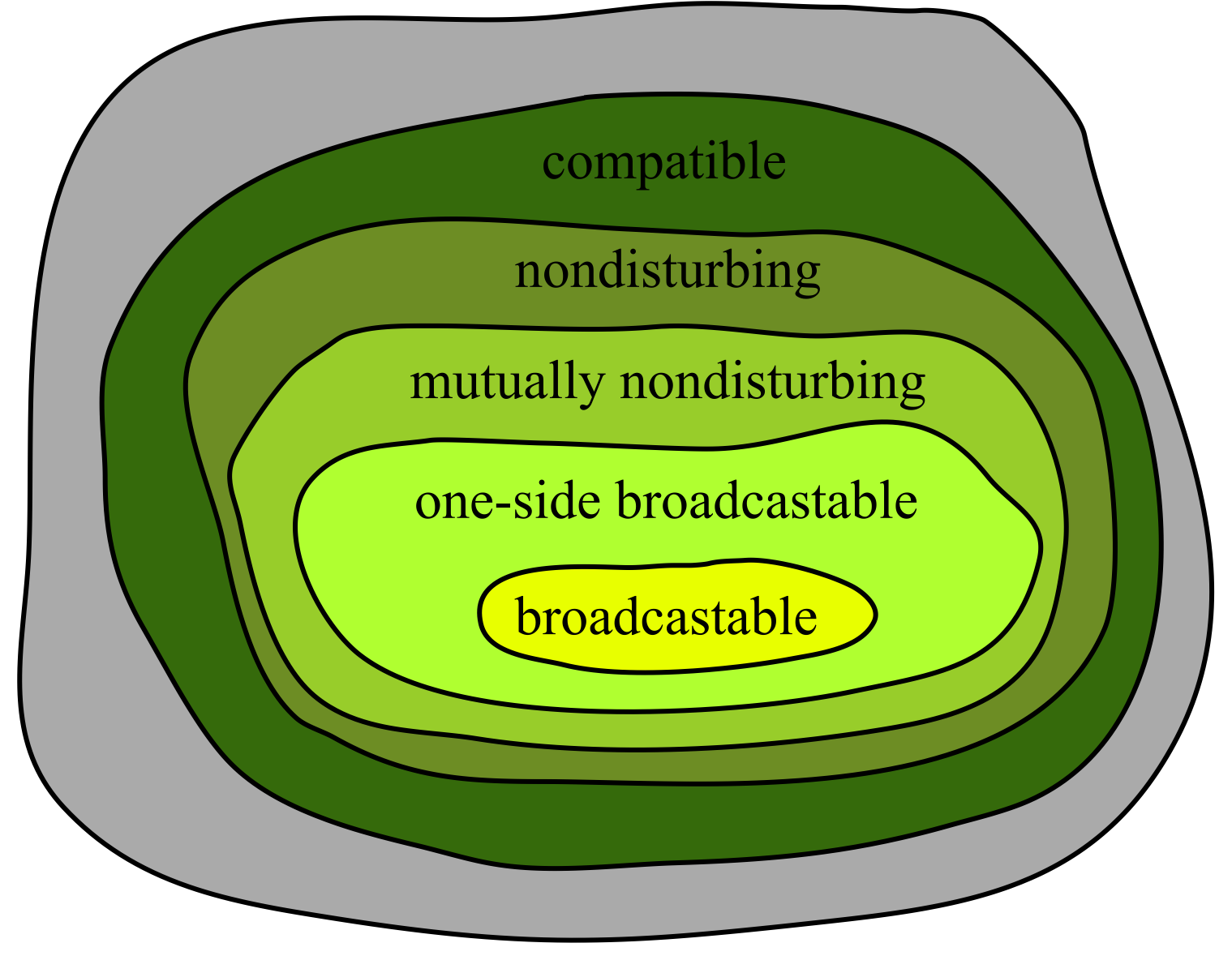}
\end{center}
\caption{\label{fig:hierarchy} The whole area depicts the set of all pairs of quantum observables. The strictest condition for a pair of observables is broadcastability, and the loosest is compatibility. The other three properties are between these two.}
\end{figure}

The most general formulation of simultaneous measurability is based on the concept of a joint observable \cite{Busch87}, \cite{LaPu97}.
A \emph{joint observable} of two observables $\A$ and $\B$ is an observable $\J:\Omega_\A \times \Omega_\B \to \lh$ such that
\begin{align}\label{eq:marginals}
& \J(x,\Omega_\B) = \A(x) \, , 
&  \J(\Omega_\A,y) = \B(y)
\end{align}
for all $x\in\Omega_\A,y\in\Omega_\B$, where we have used the shorthand notation $\J(x,\Omega_\B)=\sum_{y\in\Omega_\B} \J(x,y)$ and $\J(\Omega_\A,y)=\sum_{x\in\Omega_\A} \J(x,y)$.
The existence of a joint observable determines the following symmetric relation on the set of observables.

\begin{definition}
Two observables $\A$ and $\B$ are called \emph{compatible} or \emph{jointly measurable} if they have a joint observable; otherwise they are \emph{incompatible}.
\end{definition}

In the following sections we formulate and study four symmetric relations on the set of observables that are related to simultaneous measurability of two observables, and which are all stronger than compatibility; see Fig. \ref{fig:hierarchy}.
Hence, they correspond to stronger and weaker levels  for two observables to be simultaneously measurable.
Their complement relations refer to the impossibility of simultaneous measurement using specified resources.

\section{Broadcasting and one-side broadcasting}

A \emph{quantum channel} $\Lambda$ is a completely positive linear map from an input state space $\state$ to an output state space $\stateout$.
In the following we will consider quantum channels that take a single system as an input and give two similar systems as outputs, so that $\hi'= \hi_{\mathcal{A}}\otimes\hi_{\mathcal{B}}$ and $\hi_{\mathcal{A}}=\hi_{\mathcal{B}}=\hi$.
These kind of channels are called \emph{broadcasting channels}. 
A broadcasting channel $\Lambda:\state \to \stateab$ \emph{broadcasts a state $\varrho$} if the reduced states of the output state $\Lambda(\varrho)$ coincide with the input state, i.e.,
\begin{align}
& \ptr{\mathcal{B}}{\Lambda(\varrho)} = \varrho \, , \quad \ptr{\mathcal{A}}{\Lambda(\varrho)} = \varrho \, . \label{eq:broad}
\end{align}
A subset $\mathcal{T}$ of states is \emph{broadcastable} if there is a channel $\Lambda$ that broadcasts each state $\varrho$ belonging to $\mathcal{T}$. 
It is known that a subset $\mathcal{T}$ is broadcastable if and only if all the states in $\mathcal{T}$ commute with each other \cite{Barnumetal96}, \cite{Fanetal14}.

The broadcasting conditions in \eqref{eq:broad} for a state $\varrho$ are equivalent to the requirement that the equations
\begin{align}
\tr{\varrho \A(x) }= \tr{\Lambda(\varrho) \A(x)\otimes\id} = \tr{\Lambda(\varrho) \id \otimes\A(x)} \label{eq:broad-tr}
\end{align}
hold for all observables $\A$ and outcomes $x\in\Omega_\A$.
This formulation allows us to change the aim of the broadcasting procedure; we may want to satisfy these equations for all states but only for some chosen observables.
Hence, we arrive to the following definition. 

\begin{definition}\label{def:broadcast}
A channel $\Lambda$ \emph{broadcasts an observable $\A$} if  the condition \eqref{eq:broad-tr} holds for all states $\varrho\in\state$.
A subset $\mathcal{A}$ of observables is \emph{broadcastable} if there is a channel $\Lambda$ that broadcasts every observable $\A\in\mathcal{A}$. 
\end{definition}

The requirement that the equations in \eqref{eq:broad-tr} hold for all states $\varrho\in\state$ is equivalent to the condition 
\begin{align}
\A(x) = \Lambda^*( \A(x)\otimes\id )= \Lambda^*( \id \otimes\A(x) ) \, , 
\end{align}
where $\Lambda^*$ is the dual channel of $\Lambda$.
We will mostly use the Schr\"odinger picture of $\Lambda$ and the condition \eqref{eq:broad-tr} to make the physical content more visible, but the Heisenberg picture $\Lambda^*$ is useful when we write joint observables.

The idea of concentrating in observables rather than states was presented in \cite{FeGaPa06} and further investigated in \cite{FePa07},\cite{AlSaLaSo14}.
In these works the cloning of an observable was identified with the cloning of its mean value, so our definition is slightly different from that.
However, the essential fact that cloning of observables is more related to joint measurement than is cloning of states was observed already in \cite{FeGaPa06}.

Let us then focus on the broadcastability of two observables.
By Def. \ref{def:broadcast}, a channel $\Lambda$ broadcasts two observables $\A$ and $\B$ if
\begin{align}
& \tr{\varrho \A(x) }= \tr{\Lambda(\varrho) \A(x)\otimes\id} = \tr{\Lambda(\varrho) \id \otimes\A(x)} \label{eq:broad-A} \\
& \tr{\varrho \B(y) }= \tr{\Lambda(\varrho) \B(y)\otimes\id} = \tr{\Lambda(\varrho) \id \otimes\B(y)} \label{eq:broad-B}
\end{align}
for all states $\varrho\in\state$ and outcomes $x\in\Omega_\A$, $y\in\Omega_\B$.
We can think the broadcastibility of two observables in the following way.
Two approximate copies are made of an unknown initial state $\varrho$.
One copy is sent to Alice and another one to Bob.
Both Alice and Bob can choose if they want to measure either $\A$ or $\B$ on their respective copies.
The conditions \eqref{eq:broad-A}--\eqref{eq:broad-B} guarantee that the measurement outcome probabilities are the same as in separate measurements of $\A$ and $\B$ on the initial state $\varrho$.

To provide an example of broadcastable pairs of observables, we consider the following special class of  observables.

\begin{definition}
Let $\{\varphi_j\}_{j=1}^d$ be an orthonormal basis.
An observable $\A$ is \emph{diagonal in $\{\varphi_j\}_{j=1}^d$} if 
\begin{equation}\label{eq:commutative}
\A(x) = \sum_{j =1}^d \alpha_j(x) \kb{\varphi_j}{\varphi_j} \, , 
\end{equation}
where $0\leq \alpha_j(x) \leq 1$ and $\sum_x \alpha_j(x)=1$ for all $j=1,\ldots,d$.
\end{definition}

The observable $\A$ defined in \eqref{eq:commutative} is \emph{commutative}, i.e., $\A(x)\A(y)=\A(y)\A(x)$ for all  $x,y\in\Omega_\A$.
If the dimension of $\hi$ is finite, then a commutative observable is diagonal in some orthonormal basis. 
However, if the dimension of $\hi$ is infinite, then not all commutative observables are of the form \eqref{eq:commutative} since a positive operator need not have a pure point spectrum.
We also observe that two observables $\A$ and $\B$ that are diagonal in the same basis are \emph{mutually commuting}, i.e., $\A(x)\B(y)=\B(y)\A(x)$ for all $x\in\Omega_\A,y\in\Omega_\B$.

The following observation is analogous to the fact that a set containing two commuting states is broadcastable.

\begin{proposition}\label{prop:diagonal}
Let $\mathcal{A}$ be a set of observables that are diagonal in the same orthonormal basis $\{\varphi_j\}_{j=1}^d$. 
Then $\mathcal{A}$ is broadcastable.
\end{proposition}

\begin{proof}
We define a channel $\Lambda$ as
\begin{equation}
\Lambda(\varrho) = \sum_{j=1}^d \ip{\varphi_j}{\varrho \varphi_j} \kb{\varphi_j \otimes \varphi_j}{\varphi_j \otimes \varphi_j} \, .
\end{equation}
If $\A$ has the form \eqref{eq:commutative}, then
\begin{equation*}
\tr{\varrho \A(x) }= \tr{\Lambda(\varrho) \A(x)\otimes\id} = \tr{\Lambda(\varrho) \id \otimes\A(x)} \, ,
\end{equation*}
hence $\Lambda$ broadcasts $\A$.
\end{proof}

In a finite dimensional Hilbert space a commutative set of selfadjoint operators can be diagonalized in the same orthonormal basis.
The following statement is hence a direct consequence of Prop. \ref{prop:diagonal}.

\begin{proposition}\label{prop:commu}
Let $\dim\hi < \infty$.
A mutually commuting pair of commutative observables is broadcastable.
\end{proposition}

A relaxation of the broadcasting conditions \eqref{eq:broad-A}--\eqref{eq:broad-B} is that we require only
\begin{align}
& \tr{\varrho \A(x) } = \tr{\Lambda(\varrho) \A(x)\otimes\id} \label{eq:a-ub}\\
& \tr{\varrho \B(y) } = \tr{\Lambda(\varrho) \id \otimes\B(y)} \label{eq:b-ub}
\end{align}
for all states $\varrho\in\state$ and outcomes $x\in\Omega_\A$, $y\in\Omega_\B$.
This still refers to a process where we first make approximate copies of $\varrho$ by using the channel $\Lambda$ and then measure $\A$ and $\B$ on those copies; see Fig. \ref{fig:boxes}a.
The difference to the earlier broadcasting set-up is that now the sides of the measurements are relevant: Alice must measure $\A$ and Bob must measure $\B$ on their respective subsystems.
We are led to the following definition.

\begin{definition}\label{def:one-sided}
Two observables $\A$ and $\B$ are \emph{one-side broadcastable} if there exists a channel $\Lambda:\state \to \statehh$ such that \eqref{eq:a-ub}--\eqref{eq:b-ub} hold for all states $\varrho\in\state$ and outcomes $x\in\Omega_\A$, $y\in\Omega_\B$.
\end{definition}

It is clear that if two observables are broadcastable, then they are one-side broadcastable.
Further, a one-side broadcastable pair is compatible;  if $\A$ and $\B$ are one-side broadcastable with a channel $\Lambda$, then they have a joint observable $\J$ defined as
\begin{align}
\J(x,y) = \Lambda^* (\A(x) \otimes \B(y) ) \, .
\end{align}
The conditions \eqref{eq:a-ub} -- \eqref{eq:b-ub} guarantee that $\J$ is indeed a joint observable.

\begin{figure}
    \centering
    \subfigure[]
    {
         \includegraphics[width=6cm]{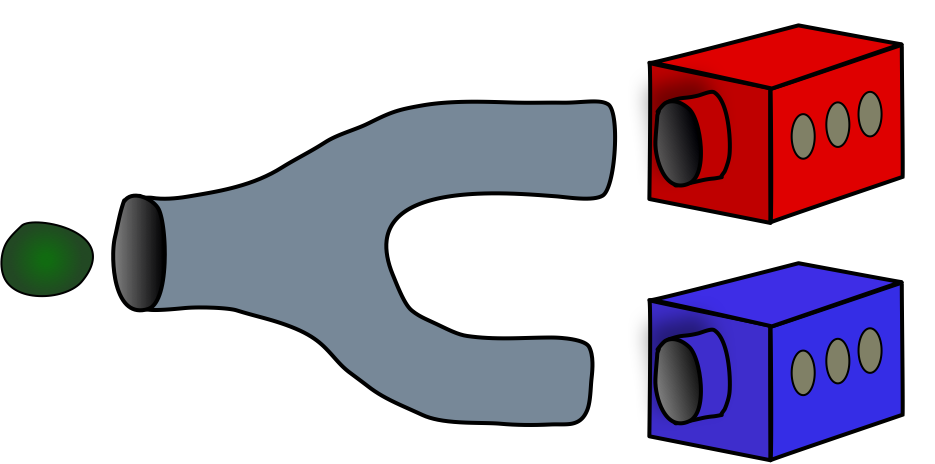}      
         }
    \subfigure[]
    {
        \includegraphics[width=6cm]{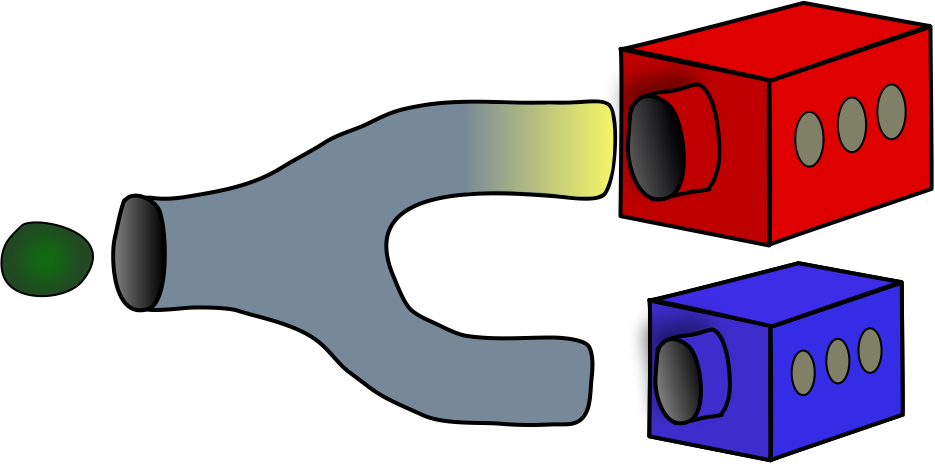} 

    }
    \subfigure[]
     {
        \includegraphics[width=6cm]{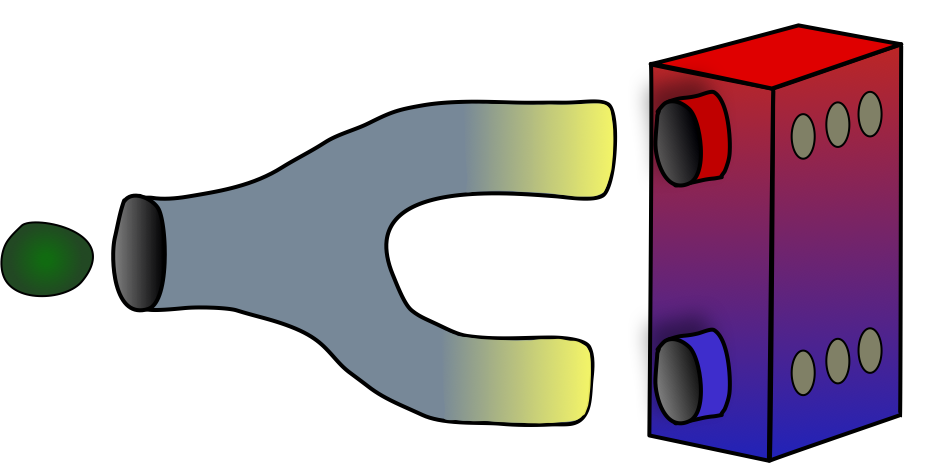} 

    }

        \caption{ \label{fig:boxes}(a) In the one-side broadcasting scenario two approximate copies of the input state are produced and the target observables are measured on these copies. (b) In a nondisturbing sequential measurement one of the measured observables is allowed to be different than the corresponding target observable. The auxiliary observable can operate on a different Hilbert space than the target observable.
        (c) In the most general set-up a global measurement is allowed. Two observables can be obtained in this way exactly when they are compatible.}
\end{figure}

\section{Nondisturbing measurements}

An observable $\A$ can be measured without disturbing another observable $\B$ if the measurement outcome distributions of $\B$ are the same if we measure $\A$ before $\B$ or not measure $\A$ at all.
To formulate this relation in the standard mathematical formalism, we recall the concept of an instrument \cite{QTOS76}.
An instrument which implements a measurement of $\A$ is a map $x\mapsto \I_x$ such that each $\I_x$ is a completely positive linear map and 
\begin{equation}\label{eq:i-repro}
\tr{\I_x(\varrho)} = \tr{\varrho \A(x)}
\end{equation}
for all states $\varrho$ and outcomes $x\in\Omega_\A$.
We will again use the notation $\I_X \equiv \sum_{x\in X} \I_x$ for all subsets $X\subseteq\Omega_\A$.
The nondisturbance condition for an observable $\B$ then reads
\begin{equation}\label{eq:i-nond}
\tr{\varrho \B(y) } =  \tr{\I_\Omega(\varrho) \B(y)}\, , 
\end{equation}
required to hold for all states $\varrho\in\sh$ and outcomes $y\in\Omega_\B$.
We say that an observable \emph{$\A$ can be measured without disturbing $\B$} if there exists an instrument $\I$ such that \eqref{eq:i-repro}--\eqref{eq:i-nond} hold for all states $\varrho\in\sh$ and outcomes $x\in\Omega_\A$, $y\in\Omega_\B$.

If $\A$ can be measured without disturbing $\B$, then $\A$ and $\B$ are compatible.
This is clear since a sequential measurement of $\A$ followed by $\B$ is a joint measurement of $\A$ and $\B$ if the first measurement does not disturb $\B$.
A joint observable $\J$ is defined as $\J(x,y) = \I_x^*(\B(y))$ and the marginal conditions \eqref{eq:marginals} follow from \eqref{eq:i-repro}--\eqref{eq:i-nond}.

To see a connection to the one-side broadcasting, we recall that every instrument can be written in the measurement model form 
\begin{equation*}
\I_x(\varrho) = \ptr{\hik}{U \eta \otimes \varrho U^* \A'(x) \otimes \id} \, , 
\end{equation*}
where $\eta$ is a fixed initial state of an ancillary system $\hik$, $\A'$ is a probe observable on $\hik$ and $U:\hik\otimes\hi \to \hik\otimes\hi$ is a unitary operator describing a measurement interaction \cite{Ozawa84}.
The condition \eqref{eq:i-repro} can then be written as
\begin{equation}
\tr{\varrho \A(x)} = \tr{U \eta \otimes \varrho U^* \A'(x) \otimes \id} \, ,
\end{equation}
and the nondisturbance condition \eqref{eq:i-nond} takes the form
\begin{align}
\tr{\varrho \B(y) } = \tr{ U \eta \otimes \varrho U^* \id \otimes \B(y) } \, .
\end{align}
By denoting $\Lambda(\varrho) = U \eta \otimes \varrho U^*$ we can write these equations as
\begin{align}
& \tr{\varrho \A(x) } = \tr{\Lambda(\varrho) \A'(x)\otimes\id} \label{eq:repro}\\
& \tr{\varrho \B(y) } = \tr{\Lambda(\varrho) \id \otimes\B(y)} \label{eq:nond}
\end{align}
These are exactly the same equations as in one-side broadcasting, except that in the latter case it is required that $\hik=\hi$ and $\A'=\A$.
This difference is illustrated in Fig. \ref{fig:boxes}b.

To see that the previous conditions are equivalent to the existence of a nondisturbing measurement, assume there is a channel $\Lambda:\state \to \statekh$ and an observable $\A'$ on $\hik$ such that  \eqref{eq:repro}--\eqref{eq:nond} hold for all states $\varrho$ and outcomes $x\in\Omega_\A$, $y\in\Omega_\B$.
For each $x\in\Omega_\A$, we then define a map $\I_x$ as
\begin{align}
\I_x(\varrho) = \ptr{\hik}{\sqrt{\A'(x)}\otimes\id \Lambda(\varrho) \sqrt{\A'(x)}\otimes\id} \, .
\end{align}
As $\I_x$ is a composition of completely positive maps, it is completely positive.
A direct calculation shows that $\I$ satisfies \eqref{eq:i-repro}--\eqref{eq:i-nond}.
We summarize the previous discussion in the following proposition.

\begin{proposition}\label{prop:nondist}
An observable $\A$ can be measured without disturbing an observable $\B$ if and only if there exists an ancillary system $\hik$, a probe observable $\A'$ on $\hik$, and a channel $\Lambda:\state \to \statekh$ such that 
\begin{align}
& \tr{\varrho \A(x) } = \tr{\Lambda(\varrho) \A'(x)\otimes\id} \\
& \tr{\varrho \B(y) } = \tr{\Lambda(\varrho) \id \otimes\B(y)}
\end{align}
hold for all  states $\varrho$ and outcomes $x\in\Omega_\A$, $y\in\Omega_\B$.
\end{proposition}

We are interested on symmetric relations on the set of observables, hence we make the following definitions.

\begin{definition}
Two observables $\A$ and $\B$ are
\begin{itemize}
\item \emph{mutually nondisturbing} if $\A$ can be measured without disturbing $\B$ and $\B$ can be measured without disturbing $\A$.
\item \emph{nondisturbing} if $\A$ can be measured without disturbing $\B$ or $\B$ can be measured without disturbing $\A$.
\end{itemize}
\end{definition}

If $\A$ can be measured without disturbing $\B$, it does \emph{not} imply that $\B$ can be measured without disturbing $\A$.
An example demonstrating this fact was given in \cite{HeWo10}.
We thus conclude that the mutual nondisturbance is a strictly stronger relation than the disturbance relation.
As we noted earlier, nondisturbing observables are compatible.
Further, a comparison of Prop. \ref{prop:nondist} with Def. \ref{def:one-sided} shows that one-side broadcastable observables are mutually nondisturbing.
We have thus reached the hierarchy depicted in  Fig. \ref{fig:hierarchy}.

As a demonstration, let us recall a class of mutually nondisturbing pairs of observables: \emph{two mutually commuting observables $\A$ and $\B$ are mutually nondisturbing} \cite{BuSi98}.
This can be seen by using the L\"uders instruments of $\A$ and $\B$.
The L\"uders instrument of $\A$ is defined as
\begin{equation}
\I_x(\varrho) = \sqrt{\A(x)} \varrho \sqrt{\A(x)} \, .
\end{equation}
It follows from $\A(x)\B(y)=\B(y)\A(x)$ that $\sqrt{\A(x)}\B(y)=\B(y)\sqrt{\A(x)}$.
Hence,
\begin{align*}
\tr{\I_{\Omega}(\varrho) \B(y)} &= \sum_x  \tr{\sqrt{\A(x)} \varrho \sqrt{\A(x)}\B(y)} \\
&= \sum_x  \tr{\varrho \B(y) \A(x)} = \tr{\varrho \B(y)} \, ,
\end{align*}
so that $\A$ can be measured without disturbing $\B$.

\section{Reformulation of compatibility}

All the four relations stronger than compatibility have been formulated as certain requirements on a broadcasting channel and auxiliary observables.
We will now put the compatibility relation into this same framework.

Let us look a relaxation of the nondisturbance as it was formulated in Prop. \ref{prop:nondist}.
We can ask for the existence of two ancillary systems $\hik_1$, $\hik_2$, a channel $\Lambda:\state \to \statekk$, and observables $\A'$ and $\B'$ on systems $\hik_1$ and $\hik_2$, respectively, such that
\begin{align}
& \tr{\varrho \A(x) } = \tr{\Lambda(\varrho) \A'(x)\otimes\id} \label{eq:gen-a}\\
& \tr{\varrho \B(y) } = \tr{\Lambda(\varrho) \id \otimes\B'(y)}  \label{eq:gen-b}
\end{align}
for all states $\varrho$ and outcomes $x\in\Omega_\A$, $y\in\Omega_\B$.
This is a relaxation of the nondisturbance relation as now auxiliary observables are allowed on both sides of the output.
We can go one step further and ask for the existence of a channel $\Lambda:\state\to\stateout$ and observable $\G$ on an arbitrary output space $\hout$ such that
\begin{align}
& \tr{\varrho \A(x) } = \tr{\Lambda(\varrho)  \G(x,\Omega_\B)} \label{eq:global-a}\\
& \tr{\varrho \B(y) } = \tr{\Lambda(\varrho)  \G(\Omega_\A,y) } \label{eq:global-b}
\end{align}
for all states $\varrho$ and outcomes $x\in\Omega_\A$, $y\in\Omega_\B$.
This includes the case when $\hout=\hik_1 \otimes\hik_2$ and $\G$ is a global observable; see Fig. \ref{fig:boxes}c.
Both of the above generalizations are equivalent to the compatibility; this is the content of the following result.

\begin{proposition}\label{prop:comp}
For two observables $\A$ and $\B$, the following are equivalent:
\begin{itemize}
\item[(i)] $\A$ are $\B$ compatible. 
\item[(ii)] There exist ancillary systems $\hik_1$ and $\hik_2$, probe observables $\A'$ on $\hik_1$ and $\B'$ on $\hik_2$, and a channel $\Lambda:\state \to \statekk$ such that \eqref{eq:gen-a}--\eqref{eq:gen-b} hold for all states $\varrho$ and outcomes $x\in\Omega_\A$, $y\in\Omega_\B$.
\item[(iii)] There exists an ancillary system $\hik$, a channel $\Lambda$ and an observable $\G$ such that \eqref{eq:global-a}--\eqref{eq:global-b} hold for all states $\varrho$ and outcomes $x\in\Omega_\A$, $y\in\Omega_\B$.
\end{itemize}
\end{proposition}

\begin{proof}
We have (iii)$\Rightarrow$(i) as $\J(x,y) = \Lambda^*(\G(x,y))$ defines a joint observable of $\A$ and $\B$.
It is clear that (ii)$\Rightarrow$(iii) since there are less constrains in (iii) than in (ii). 
To see that (i)$\Rightarrow$(ii), assume that $\A$ are $\B$ compatible, so there exists a joint observable $\J$.
We fix Hilbert spaces $\hik_1$ and $\hik_2$ with the dimensions $\# \Omega_\A$ and $\# \Omega_\B$, respectively.
On both of these Hilbert spaces we fix orthonormal bases $\{\varphi_x\}$ and $\{\eta_y\}$, labeled with the elements of $\Omega_\A$ and $\Omega_\B$.
We then define a channel $\Lambda$ as
\begin{equation}
\Lambda(\varrho) = \sum_{x,y} \tr{\varrho \J(x,y)} \kb{\varphi_x \otimes \eta_y}{\varphi_x \otimes \eta_y} \, ,
\end{equation}
and we define the observables $\A'$ and $\B'$ as
\begin{equation}
\A'(x) = \kb{\varphi_x}{\varphi_x} \, , \quad \B'(y) = \kb{\eta_y}{\eta_y} \, .
\end{equation}
With these choices the requirements \eqref{eq:gen-a}--\eqref{eq:gen-b} are satisfied.
\end{proof}

\section{Qualitative differences}\label{sec:qualitative}

The qualitative differences of the two extreme relations, broadcasting and compatibility, to the other relations link to the fundamental theorems of no-broadcasting \cite{Barnumetal96} and no-information-without-disturbance \cite{Busch09}.
In the following we explain these connections, which are both based on the concept of an informationally complete observable.
By definition, a collection $\mathcal{A}$ of observables is \emph{informationally complete} if the measurement data $\{ \tr{\varrho \A(x)} : \A\in\mathcal{A},x\in\Omega_\A\}$ is unique for every state $\varrho\in\state$ \cite{BuLa89}.
Even a single observable can be informationally complete \cite{SiSt92}, and a standard example of such observable is a covariant phase space observable (either in finite or infinite phase space) satisfying certain criterion \cite{KiLaScWe12}.

\subsection{Broadcastability versus other relations}\label{sec:versus}

The no-broadcasting theorem for states implies some immediate limitations on the broadcastability of subsets of observables.
Namely, let $\mathcal{A}$ be an informationally complete set of observables.
The broadcastability of $\mathcal{A}$ would then imply that the reduced states of the bipartite output state $\Lambda(\varrho)$ coincide with the input state $\varrho$.
This cannot hold for all states by the no-broadcasting theorem, so we conclude that \emph{an informationally complete set of observables is not broadcastable}.
In particular, a single informationally complete observable is not broadcastable.

The formulation of the broadcastability relation implies a trivial but significant feature: if two observables $\A$ and $\B$ are broadcastable, then $\A$ is broadcastable with itself. 
Therefore, an informationally complete observable is not broadcastable with any other observable.
In the language of binary relations, this means that informationally complete observables are isolated elements in the broadcasting relation.

The existence of isolated elements, i.e., observables that are not related to any other observable, is a distinctive feature of the broadcasting relation.
Too see this, we observe that \emph{every observable is one-side broadcastable with any trivial observable}.
By a trivial observable we mean an observable for which the measurement outcome probabilities do not depend on the input state at all.
Mathematically, this kind of observable can be written as $\T(x) = t(x) \id$, where $t$ is a probability distribution and $\id$ is the identity operator.
Hence, to prove the claim, fix a state $\eta\in\sh$ and define a channel $\Lambda$ as $\Lambda(\varrho) = \varrho \otimes \eta$.
Let $\A$ be any observable and $\T$ a trivial observable. 
Then 
\begin{align}
&  \tr{\Lambda(\varrho) \A(x)\otimes\id} = \tr{\varrho \A(x) } \\
& \tr{\Lambda(\varrho) \id \otimes\T(y)}  = \tr{\eta \T(y) } = \tr{\varrho \T(y) }
\end{align}
so $\A$ and $\T$ are one-side broadcastable.
Due to the hierarchy of the relations, we conclude that a trivial observable is related to any other observable in all the relations except broadcasting.

\subsection{Compatibility versus other relations}

A specific feature of the compatibility relation is that \emph{every observable is compatible with itself}. 
To see this, let $\A$ be an observable. We define an observable $\J$ on $\Omega_\A\times\Omega_\A$ as $\J(x,y) = \delta_{xy} \A(x)$.
Then 
\begin{align}
\J(x,\Omega_\A) =  \J(\Omega_\A,x) = \A(x)
\end{align}
for all $x \in \Omega_\A$, hence $\J$ is a joint observable of $\A$ and $\A$.
The physical explanation of this feauture is the fact that the measurement outcomes of $\A$ are distinguishable classical states and therefore can be duplicated. 

This reflexivity of the compatibility relation is a qualitative difference to all other four relations: in all of them we can find an observable $\A$ which is not in the given relation with itself. 
Due to the hierarchy of the relations, it is enough to find an observable $\A$ such that $\A$ cannot be measured without disturbing $\A$ itself.
A whole class of these kind of observables consists of informationally complete observables. 
The no-information-without-disturbance theorem states that a measurement that gives information causes necessarily some disturbance. 
Since, by definition, an informationally complete observable $\A$ gives unique probability outcome distribution to all states, we conclude that a measurement of $\A$ necessarily disturbs a subsequent measurement of $\A$.

Another distinctive feature, more important but not as sharply formulated, is the fact that \emph{addition of sufficient amount of noise makes any pair of observables compatible} \cite{BuHeScSt13}, \cite{HeMiZi16}.
By the addition of noise we mean mixing an observable with a trivial observable.
For instance, let us consider two incompatible observables $\A$ and $\B$ and their deformations $\tilde{\A}$ and $\tilde{\B}$, where 
\begin{equation}\label{eq:deformed}
\tilde{\A}(x) = \half \A(x) + \half t_1 (x) \id \, , \quad \tilde{\B}(y) = \half \B(y) + \half t_2 (y) \id
\end{equation}
and $t_1,t_2$ are some probability distributions on $\Omega_\A,\Omega_\B$, respectively.
Then $\tilde{\A}$ and $\tilde{\B}$ are compatible as they have a joint observable
\begin{equation}
\J(x,y) = \half t_2(y) \A(x) + \half t_1(x) \B(y) \, .
\end{equation}
Using the normalizations $\sum_x \A(x)=\sum_y \B(y)=\id$ and $\sum_x t_1(x) = \sum_y t_2(y) =1$ it is straightforward to verify that $\J$ gives $\tilde{\A}$ and $\tilde{\B}$ as its marginals.

In contrast, \emph{addition of white noise does not make an arbitrary pair of observables nondisturbing}.
To see this, suppose that $\A$ is informationally complete.
Then a deformed observable of the form
\begin{equation}\label{eq:deformed-2}
\widetilde{\A}(x) = \lambda \A(x) + (1-\lambda) t(x) \id
\end{equation}
with $0<\lambda \leq 1$ is still informationally complete.
This follows from the fact that an observable is informationally complete if and only if its range spans $\lh$ \cite{SiSt92}, and the deformation in \eqref{eq:deformed-2} does not change the span of the range when $\lambda \neq 0$.
Therefore, our earlier discussion implies that $\widetilde{\A}$ cannot be measured without disturbing itself.

\section{Qubit observables}

As we noted earlier, a nondisturbing pair of observables need not be mutually nondisturbing.
However, if the dimension of the Hilbert space is $2$, then these relations are the same and equivalent to the mutual commutativity.
Namely, the result \cite[Prop. 6]{HeWo10} implies the following:

\begin{proposition}\label{prop:qubit-1}
For two qubit observables $\A$ and $\B$, the following are equivalent:
\begin{itemize}
	\item[(i)] $\A$ and $\B$ are mutually commuting.
	\item[(ii)] $\A$ and $\B$ are mutually nondisturbing.
	\item[(iii)] $\A$ and $\B$ are nondisturbing.
\end{itemize}
\end{proposition}

Using Prop. \ref{prop:commu} and Prop. \ref{prop:qubit-1} we get a complete characterization of broadcastable pairs of qubit observables.

\begin{proposition} \label{prop:qubit-2}
Two qubit observables $\A$ and $\B$ are broadcastable if and only if $\A$ and $\B$ are commutative and mutually commuting.  
\end{proposition}

\begin{proof}
The 'if' part is a direct consequence of Prop. \ref{prop:commu}.
To show the 'only if' part, we assume that $\A$ and $\B$ are broadcastable.
Then $\A$ and $\B$ are also mutually nondisturbing, hence by Prop. \ref{prop:qubit-1} mutually commuting.
Further, since $\A$ and $\B$ are broadcastable,  $\A$ is broadcastable with itself.
By the hierarchy of the relations this implies that $\A$ is nondisturbing  with itself, hence using again Prop. \ref{prop:qubit-1} we conclude that $\A$ is commutative.
In a similar way we conclude that $\B$ is commutative.
\end{proof}

Further, utilizing the hierarchy of relations and the previous results, we can also characterize the one-side broadcastability of qubit observables.
The following statement extends Prop. \ref{prop:qubit-1}.

\begin{proposition}\label{prop:qubit-final}
For two qubit observables $\A$ and $\B$, the following are equivalent:
\begin{itemize}
	\item[(i)] $\A$ and $\B$ are one-side broadcastable.
	\item[(ii)] $\A$ and $\B$ are mutually nondisturbing.
	\item[(iii)] $\A$ and $\B$ are nondisturbing.
	\item[(iv)] $\A$ and $\B$ are mutually commuting.
\end{itemize}

If one (and hence all) of these relations holds and neither $\A$ nor $\B$ is trivial, then $\A$ and $\B$ are broadcastable.
\end{proposition}

\begin{proof}
By the general hierarchy of the relations we have (i)$\Rightarrow$(ii)$\Rightarrow$(iii), and by Prop. \ref{prop:qubit-1} we have (iii)$\Leftrightarrow$(iv). It is thus enough to show that (iv)$\Rightarrow$(i).
Let $\A$ and $\B$ be mutually commuting qubit observables.
Then at least one of the following holds:
  \begin{itemize}
	\item[(a)] $\A$ and $\B$ are both commutative.
	\item[(b)] $\A$ is a trivial observable.
	\item[(c)] $\B$ is a trivial observable.
\end{itemize}
To see this, let us first note that a selfadjoint operator on a two-dimensional Hilbert space either has nondegenerate spectrum or is a multiple of the identity operator.
Now, assume that the observable $\A$ is not commutative, and let $\A(x)$ and $\A(x')$ be two noncommuting operators.
Since $\B(y)$ commutes with both $\A(x)$ and $\A(x')$, it is diagonal in the eigenbases of $\A(x)$ and $\A(x')$.
It follows that $\B(y)$ is a multiple of the identity operator.
Therefore, if $\A$ is not commutative, then $\B$ is trivial, and vice versa.

The one-side broadcastability of $\A$ and $\B$ follows in all cases (a)--(c).
If (a) holds, then by Prop. \ref{prop:qubit-2} the pair is broadcastable, hence one-side broadcastable.
If (b) or (c) holds, then the pair is one-side broadcastable since we seen in Sec. \ref{sec:versus} that every observable is one-side broadcastable with any trivial observable.

The last claim follows from the division into the cases (a)--(c) and Prop. \ref{prop:qubit-2}.
\end{proof}

We recall that two qubit observables can be compatible even if they are not mutually commuting.
For instance, the compatibility relation for the pairs of two-outcome qubit observables has been characterized in \cite{StReHe08},\cite{BuSc10},\cite{YuLiLiOH10}, and it is easy to see that most of the compatible pairs are not mutually commuting.

\section{Discussion}

The set of bipartite states divides into separable states and entangled states.
Among all separable states, some states are more classical than others.
Especially, the set of zero discord states is a proper subset of separable states, and separable states with nonzero discord yield advantage over zero discord states in certain tasks like phase estimation \cite{Girolamietal14}.

A comparable partitioning on pairs of observables is the division into compatible pairs and incompatible pairs, and then compatible pairs further into  subsets of broadcastable, one-side broadcastable, nondisturbing and mutually nondisturbing pairs.
It would be interesting to see if their complement relations have a similar kind of task oriented characterizations as incompatibility, in which case a pair is incompatible if and only if it enables steering \cite{UoBuGuPe15}.

\end{document}